\documentclass[journal]{IEEEtran}

\usepackage{amsthm}
\usepackage{times,amsmath,epsfig}
\usepackage[T1]{fontenc}
\usepackage{graphicx} 
\usepackage{color} 
\usepackage{subfigure}
\usepackage{float}
\usepackage{boxedminipage}
\usepackage{multirow}

\usepackage{cite}

\newtheorem{proposition}{Proposition}

\begin{document}

\title{ALOHA Random Access that Operates as a Rateless Code}

\author{\v Cedomir~Stefanovi\' c,~\IEEEmembership{Member,~IEEE,} Petar~Popovski,~\IEEEmembership{Senior Member,~IEEE}

\thanks{\v Cedomir Stefanovi\' c and Petar Popovski are with the Department of Electronic Systems, Aalborg University,
Aalborg, Denmark (email: \{cs,petarp\}@es.aau.dk).}
\thanks{The research presented in this paper was partly supported by the Danish Council for Independent Research (Det Frie Forskningsr{\aa}d) within the Sapere Aude Research Leader program, Grant No. 11-105159 ``Dependable Wireless Bits for Machine-to-Machine (M2M) Communications'' and performed partly in the framework of the FP7 project ICT-317669 METIS, which is partly funded by the European Union. The authors would like to acknowledge the contributions of their colleagues in METIS, although the views expressed are those of the authors and do not necessarily represent the project.}
}

\maketitle

\begin{abstract}
Various applications of wireless Machine-to-Machine (M2M) communications have rekindled the research interest in random access protocols, suitable to support a large number of connected devices.
Slotted ALOHA and its derivatives represent a simple solution for distributed random access in wireless networks.
Recently, a framed version of slotted ALOHA gained renewed interest due to the incorporation of  successive interference cancellation (SIC) in the scheme, which resulted in substantially higher throughputs.
Based on similar principles and inspired by the rateless coding paradigm, a frameless approach for distributed random access in slotted ALOHA framework is described in this paper. 
The proposed approach shares an operational analogy with rateless coding, expressed both through the user access strategy and the adaptive length of the contention period, with the objective to end the contention when the instantaneous throughput is maximized.
The paper presents the related analysis, providing heuristic criteria for terminating the contention period and showing that very high throughputs can be achieved, even for a low number for contending users.
The demonstrated results potentially have more direct practical implications compared to the approaches for coded random access that lead to high throughputs only asymptotically. 
\end{abstract}


\section{Introduction}
\label{sec:intro}

Slotted ALOHA \cite{R1975} is a well known distributed random access scheme in which the link time is divided into slots of equal duration and the users contend to access the Base Station (BS) by transmitting with a predefined slot-access probability $p_a$.
Framed ALOHA \cite{OIN1977} is a variant in which the link time is divided into frames containing $M$ slots, and the users contend by transmitting in a single, randomly chosen slot of the frame.
Both in slotted and framed ALOHA, only the slots containing a single user transmission (i.e., singleton slots) are useful and the corresponding transmission is successfully resolved, while the slots containing no user transmission (i.e., idle slots) or multiple user transmissions (i.e., collision slots) are wasted.
The throughput $T$, defined as the probability of successfully receiving a user transmission per slot, is equal to the probability of having a singleton slot.
In case of slotted ALOHA, the throughput is maximized when the slot-access probability is set to $p_a = 1/N$, where $N$ is the number of contending users.
Similarly, the throughput of framed ALOHA is maximized when $M=N$.
In both cases, the maximal throughput is $T_{\max} = 1/e \approx 0.37$, achieved for $N \rightarrow \infty$.

A major upgrade of the framed ALOHA was proposed in \cite{CGH2007}, where each user sends replicas of the packet in multiple slots of the frame, thus lowering the fraction of idle slots and increasing the fraction of collision slots. However, instead of only using the singleton slots, the receiver also uses the collision slots by applying successive interference cancellation (SIC).\footnote{The idea of users transmitting multiple replicas of the same data in the frame, but without SIC, was originally proposed in \cite{CR1983}.}
As a result, the throughput is substantially increased, as it now accounts both for the successfully received and subsequently resolved transmissions.
For the simplest case in which each user transmits two replicas of the same packet, the maximal throughput is increased to $T_{\max}\approx0.55$.

Another major improvement of framed ALOHA was made in \cite{L2011}, which related 
the process of successive interference cancellation applied to colliding users to the process of iterative belief-propagation (BP) erasure-decoding of codes-on-graphs. This opened the way to use the rich tools of codes-on-graphs in designing random access strategies. 
The optimal access strategy is analogous to coding of left-irregular LDPC codes: based on a predefined probability distribution, each user selects randomly and independently a number of slots from the frame in which he will repeat the replicas of the same message. 
In this way, the achievable throughput can further be increased and a suitable selection of the distribution can lead to $T_{\max} \rightarrow 1$ as $N \rightarrow \infty$ \cite{L2011,NP2012}.

Exploiting the analogies between SIC and iterative BP erasure-decoding, in this paper we elaborate the approach of \emph{frameless ALOHA} for distributed random access.
Our motivation stems from another type of codes-on-graphs with advantageous properties, namely the \emph{rateless} codes \cite{BLMR1998}.
Frameless ALOHA employes a full operational analogy with rateless codes, seen in two aspects.
First, the packet of each user acts as an information symbol, while each received slot at the BS acts as an encoded symbol. Second, in rateless coding the duration of the transmission period is not predefined, but it ends when the receiver decodes the data successfully and sends a terminating feedback to the transmitter. 
Likewise, in frameless ALOHA the the frame length is not \emph{a priori fixed} (thereby the term ``frameless''), but the BS terminates the contention process adaptively, when a termination criterion implemented at the receiver is satisfied. The results show that the adaptive termination is central to the throughput maximization.
We also note that the proposed approach, in which the complexity is transferred to the receiver at the Base Station, is highly desirable for M2M communication scenarios. 

\begin{figure*}[t]
	\begin{center}
\includegraphics[width=1.6\columnwidth]{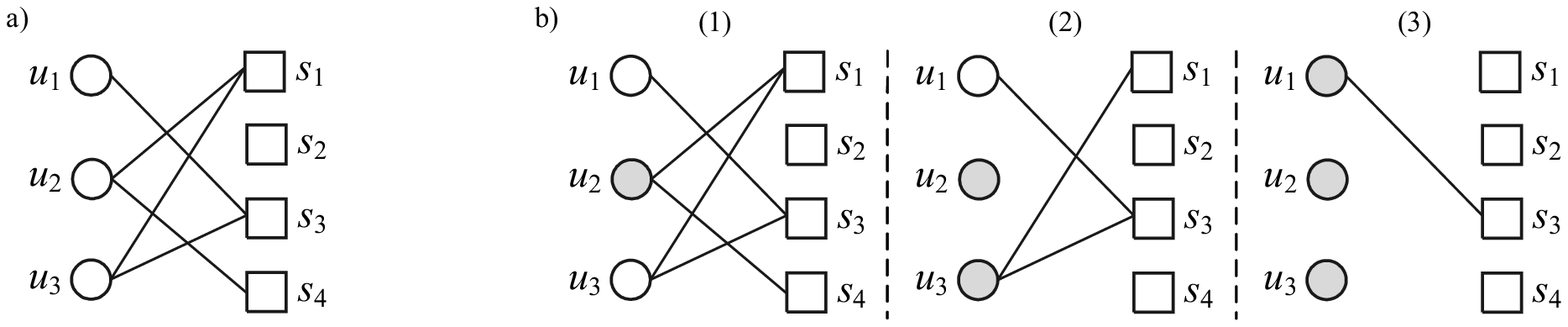}
	\end{center}
\caption{a) Graph representation of framed ALOHA b) Execution of SIC on graph.}
	\label{fig:graph}
\end{figure*}

The initial ideas and analysis of the frameless random access were presented in \cite{SPV2012}, where the performance of the scheme was investigated under varying ratios $\frac{M}{N}$.
In this paper we expand the concepts and analysis in multiple ways.
First, in contrast to \cite{SPV2012}, we investigate the performance of the scheme using the termination criterion related to the monitoring of the fraction of resolved user transmissions and instantaneous throughput.
We show that a simple heuristic approach based on the asymptotic analysis, coupled with a rather simple variant of the access strategy that is uniform both over users and over slots, grants throughputs that are the highest reported in the literature for practical number of contending users (i.e., $N \in [50,1000]$), significantly surpassing the results presented in \cite{SPV2012}.
Further, we consider the impact of the noise-induced packet erasures on the performance of the proposed scheme and, as the parameters of the scheme depend on $N$, we  analyze the impact of the accuracy of its estimate as well.
We also provide insights into the issues related to practical implementation and elaborate on the similarities/differences with the standard rateless and raptor codes.
Finally, we note that the goal of the presented analysis is the maximization of the one-shot throughput of the scheme; in the related work \cite{STPP2013}, we analyzed the overall throughput of the scheme in the case of a batch arrival of $N$ users that are contending over multiple rounds.
However, the emphasis in \cite{STPP2013} was on the design of a suitable algorithm for the estimation of $N$, and not on the optimization of the termination criterion.

The organization of the rest of the paper is as follows.
Section~\ref{sec:background} presents the background and related work.
Section~\ref{sec:model} introduces the system model and presents the related analysis.
Section~\ref{sec:threshold} elaborates the principles for terminating contention period in the proposed scheme.
The sensitivity of the scheme's performance to the accuracy of the estimated number of users is explored in Section~\ref{sec:sensitivity}.
Section~\ref{sec:discussion} considers practical aspects of the scheme, as well as relations to the standard rateless codes.
Finally, Section~\ref{sec:conclusion} concludes the paper.

\section{Background and Related Work}
\label{sec:background}

\subsection{Coded Random Access and SIC}

We start by presenting the most relevant concepts from \cite{L2011}, as it represents the closest work in the state-of-the-art. 
Fig.~\ref{fig:graph}a) depicts a toy example for the graph representation of framed ALOHA with 3 users and 4 slots; the left-side nodes represent contending users, the right-side nodes represent the slots of the frame and the edges connect the users with the respective slots in which the user transmissions take place.
All transmissions performed by a user carry the same message and a pointer to all other replicas.

The execution of SIC for the same example is depicted in Fig.~\ref{fig:graph}b).
Same as in an iterative BP erasure-decoding, in the first step singleton slots are identified, i.e., slot $s_4$ and the corresponding transmissions, i.e., the transmission of the user $u_2$, resolved.
Using pointers contained in the resolved transmissions, their replicas, i.e., the respective edges, are removed from the corresponding slots, thus potentially resulting in new singleton slots.
In the above example, the replica transmitted by $u_2$ is removed from slot $s_1$, and as a result, this slot becomes a singleton. 
SIC iterates in the same way, until there are no new singleton slots or all user transmissions have been resolved.

The presented analogy between SIC and iterative BP erasure-decoding motivated the application of the theory and tools from codes-on-graphs to improve the throughput of framed ALOHA \cite{L2011}.
Specifically, throughput maximization for framed ALOHA with SIC can be represented by minimization of the symbol-error probability for the corresponding fixed-rate erasure-correcting code, when decoded by the iterative BP algorithm. 
The main difference with the standard fixed-rate code design is expressed through the fact that, due to the constraints of ALOHA framework, the encoding is done in a distributed and uncoordinated way.
This implies that the distribution of colliding transmissions over the slots, i.e., the distribution of the edges over the right-side nodes in Fig.~\ref{fig:graph}, can not be controlled directly, but rather only statistically, through the behavior of the contending users. 
In \cite{L2011} it was shown using numerical optimization that the strategy that maximizes the throughput is analogous to the encoding of left-irregular LDPC codes: following an irregular distribution, every user randomly and independently selects a number of slots from the frame in which its replicas are going to be transmitted. 
This initial work was followed by papers that further extended the scheme by applying doubly generalized LDPC codes, where instead of transmitting replicas users transmit encoded segments of the message \cite{PLC2011,PFC2010}, derived the capacity bounds \cite{PLC2011b}, and finally, applied spatially coupled codes to the framework \cite{LPLC2012}. 

The main difference of our approach with respect to the works outlined above is that its operation is based on the concepts of rateless codes, described in the next subsection.

\subsection{Rateless Codes}

Rateless (or digital fountain) codes \cite{BLMR1998} represent a class of forward error-correction codes with capacity-approaching behavior over erasure channels that is universal and valid for an arbitrary erasure statistics.
They are \emph{rateless} due to their key operational feature, described as follows.
The code rate is not set a priori  and the encoder transmits new encoded symbols\footnote{Here symbols represent equal-length packets of bits, as in a typical erasure-coding context.} continuously (the rate decreases with every transmitted symbol).
When the receiver decodes the message it sends feedback and the transmitter stops sending symbols, which effectively (and a posteriori) establishes the code rate.

The first practical capacity-approaching version of rateless codes are the LT codes~\cite{L2002}.
LT encoding is a simple process where, for each encoded symbol, a degree $d$ is sampled from a degree distribution $\Psi(d)$, then $d$ out of $K$ information symbols from the source message are uniformly selected and bit-wise XOR-ed to produce the encoded symbol.
The fundamental element of the LT code is the design of the degree distribution $\Psi(d)$, which should be such to enable the recovery of the source message using iterative BP decoding from any set of $(1+\epsilon)K$ successfully received encoded symbols, where $\epsilon$ is a small positive value.
In \cite{L2002} it was shown that the desired operation can be achieved by selecting $\Psi(d)$ to be a robust-soliton degree-distribution. 

LT codes admit sparse-graph interpretation, similar to the one presented in Fig.~\ref{fig:graph}a), where the left/right-side nodes represent the information/encoded symbols and the graph edges reflect the process of combining information symbols into encoded symbols.
As already noted, the right degree distribution is of a robust-soliton type, while the left degree distribution asymptotically tends to the Poisson distribution due to random uniform source node sampling.
In contrast, in a random access protocol each user, acting as an information symbol, chooses its transmission instants independently of the other users.
This constrains the design space to the decentralized design of left degree distributions.
At the same time, due to the decentralized and uncoordinated access strategy, the right degree distribution (which is associated with slots) tends to a Poisson distribution, as demonstrated in Section~\ref{sec:distr}.

Another major difference between the design of the optimal strategies in LT erasure-coding and interference cancellation scenarios, stems from the fact that in the former the average degree of encoded symbols scales as $O(\log K)$, i.e., it increases logarithmically with the number of information symbols.
The same strategy is not suitable for coded random access, as it would imply that slot degrees would increase with the number of contending users, which may adversely affect the interference cancellation potential \cite{FJPSL2011,ZZ2012}.
In other words, when there are more colliding transmissions in a slot, it becomes less likely that the individual transmissions can be extracted via SIC.

A recent theoretical contribution \cite{NP2012}, shows that asymptotically, the use of truncated ideal soliton distribution at the user side achieves optimal performance in a \emph{framed} ALOHA setting, i.e., $T \rightarrow 1$ when $N \rightarrow \infty$.
Contrary to \cite{NP2012}, our approach builds-up genuinely on the rateless coding principles \cite{BLMR1998}, as it operates with no prior notion of a frame and implements an adaptive termination of the contention period.
We also note that the average slot degree of the distribution proposed in \cite{NP2012} essentially scales logarithmically with the number of users\footnote{A related remark was pointed out in \cite{PLC2012}, by noting that the number of repeated user transmissions in the frame increases as $N$ increases.} $N$, and, as already discussed, the potential application of these results could be limited by the SIC performance.

\section{System Model}
\label{sec:model}

We assume that there are $D$ users in the system, out of which a random subset of $N$ \emph{active} users is contending to access to the BS.
The start and the end of the contention period is denoted through a beacon sent by the BS.
The contention period is divided into slots of equal duration, and the users are synchronized on a slot basis.
The length of the contention period is $M$ slots; $M$ is not determined a priori and gets assigned its value dynamically, when the contention period is terminated.
Within the contention, each active user accesses the medium and transmits packets on a slot basis by using identical slot-access probability $p_a$.
This value is broadcast by the BS via beacon at the start of the contention period and is determined based on the estimation of $N$ available at BS.
We will assume at first that BS knows perfectly $N$ and discuss the consequences of violating that assumption in Section~\ref{sec:sensitivity}.
The contention period ends when the BS sends a new beacon. 

After the BS transmits the beacon, the active set of users starts a contention. Different from \cite{STPP2013}, where we have considered batch arrivals, here we assume that  each of the $D$ devices is backlogged and has always a packet to send.
However, the set of active devices that transmit in a given contention period is random, due to the following.
After the reception of the beacon, the user $u_i$ estimates its channel coefficient $h_i$ and becomes active only if $|h_i|>\tau$, where $\tau$ is a predefined threshold.
The user channel coefficients are constant in a given transmission period, but may change from one to another contention period due to fading.
Here we do not deal with the fading statistics, but we remark that it can be used to estimate the number of active users $N$.

If $u_i$ is active in a given contention period, then $u_i$ transmits:
\begin{align}
\label{eq:precoding}
X'_i = \frac{h_i^*}{|h_i|^2} X_i, \; 1 \leq i \leq N,
\end{align} 
Strictly speaking, $X_i$ represents a transmitted symbol with unit power $E[|X_i|^2]=1$, but we slightly abuse the notation and consider that $X_i$ stands for a packet that $u_i$ transmits during a given contention period.
The baseband precoding coefficient $ \frac{h_i^*}{|h_i|^2}$ implements a channel inversion.
This justifies our criterion to activate a user if his channel coefficient is above the threshold, as it effectively sets an upper limit on the transmitting power.
We assume that $h_i$ does not change during a contention period and note that several replicas of the same packet $X'_i$ may be sent in different slots belonging to the same contention period.

The BS station receives composite signal $Y_j$ in slot $s_j$:
\begin{align}
\label{eq:composite}
Y_j = \sum_{i=1}^{N} a(i,j) X_i + Z_j, \;  1 \leq j \leq M,
\end{align}
where $Z_j$ is the noise and $a(i,j)=1$ if user $i$ transmits in slot $j$, is 0 otherwise, and:
\begin{align}
P\left[a(i,j)=1 \right] = p_a, \; 1 \leq i \leq N, \; 1 \leq j \leq M.
\end{align}
Due to channel inversion, the received contribution from $u_i$ at the BS in a given slot is $X_i$. At first we will introduce the frameless access mechanism by assuming that $Z_j=0$ for all $j$ and revise this assumption in Section~\ref{sec:noise}.

Each transmission of user $u_i$, $1 \leq i \leq N$, contains a replica of the same message $X_i$; we assume that each replica contains pointers to all other replicas, as this is required for the execution of SIC.\footnote{We address the practical way of sending pointers in Section \ref{sec:practical}.}
The number of replicas transmitted by $u_i$ is denoted as user degree $|u_i|$ and the number of colliding transmissions in slot $s_j$ is denoted as slot degree $|s_j|$:
\begin{align}
|u_i|=\sum_{j=1}^{M} a(i,j), \; 1 \leq i \leq N, \\
|s_j|=\sum_{i=1}^{N} a(i,j), \; 1 \leq j \leq M.
\end{align}

The BS station is able to discern among idle ($|s_j|=0$), singleton ($|s_j|=1$) and collision slots ($|s_j| > 1$), stores all the observed slots (i.e., the received composite signals), and after each observed slot, the BS performs SIC until it finishes naturally, i.e., until there are no more users that could be resolved by the stored and currently observed slot.
This is repeated until the BS decides to terminate the contention period and send a new beacon. The key element is the choice of the termination criteria, as elaborated in the next section.

\begin{figure}[tbp]
	\begin{center}
\includegraphics[width=0.81\columnwidth]{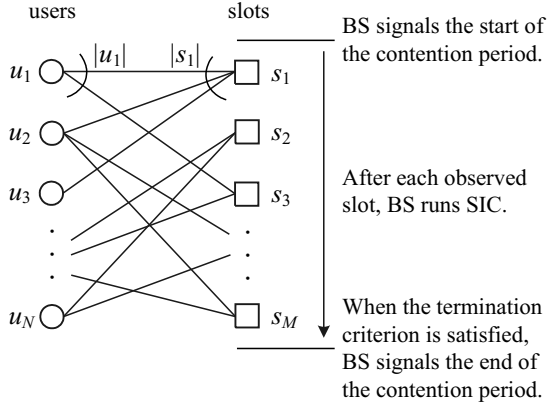}
	\end{center}
\caption{Graph representation of frameless ALOHA.}
	\label{fig:f-graph}
\end{figure}

The described framework is depicted in Fig.~\ref{fig:f-graph}.
The number of edges incident to user node $u_i$ is equal to $|u_i|$ and the number of edges incident to slot node $s_j$ is equal to $|s_j|$.
Henceforth, we use terms user/user node and slot/slot node, interchangeably.
Also, the term ``user resolution' indicates that the BS has recovered the message transmitted by the user.

\subsection{Degree Distributions}
\label{sec:distr}

We assume that all the slots in the proposed scheme have the same target degree\footnote{In a general case, the target slot degree $G$ can depend on slot number $s_j$, i.e., $G=f(s_j)$ In the initial work \cite{SPV2012}, it was shown that, despite using such a simplistic access strategy with a constant $G$, high throughputs can be achieved, comparable or even higher than the ones attainable by more involved access methods \cite{L2011}.}\textsuperscript{,}\footnote{Note that $G$ in this paper refers to number of colliding packets per slot, i.e., the average \emph{physical} slot load, whereas in \cite{L2011, PLC2011} $G$ refers to the number of the original user packets (excluding replicas) divided by the number of the slots in the frame, i.e., the average \emph{logical} slot load.} $G$.
In order to achieve $G$, the slot-access probability is set to:
\begin{align}
\label{eq:access}
p_a = \frac{G}{N},
\end{align}
It is straightforward to show that the probability that a slot $s$ is a of degree $n$ is:
\begin{align}
\label{eq:binom}
P \left[ |s| = n \right] & = \Psi_n = { N \choose n } p_a^n (1-p_a)^{N-n} \approx \frac{G^n}{n!} e^{-G},
\end{align}
using the standard approximation of the binomial distribution by the Poisson distribution (which can be assumed for the range of values for $N$ and $G$ that are of interest in this paper), and that the average slot degree is equal to the target degree $G$, i.e., $E\left[ |s| = G \right]$.
Following the standard notation used for codes-on-graphs \cite{RU2007}, the slot degree-distribution can be reduced to a particularly simple expression: 
\begin{align}
\label{eq:slot}
\Psi(x) =  \sum_{n=0}^{\infty} \Psi_n x^n = \sum_{n=0}^{\infty} \frac{G^n}{n!} e^{-G} x^n = e^{-G(1-x)}.
\end{align}

The probability of user $u$ having degree $m$ and the corresponding degree-distribution are:
\begin{align}
P \left[ |u| = m \right] & = \Lambda_m = { M \choose m } p_a^m ( 1 - p_a)^{M-m} \nonumber \\
\label{eq:puser}
& \approx \frac{\big( \left(1+\epsilon \right)G \big) ^m}{m!} e^{-( 1 + \epsilon ) G}, \\
\label{eq:user}
 \Lambda (x) & = \sum_{m=0}^{\infty} \Lambda_m x^m =  e^{-(1 + \epsilon)G(1 - x)}.
\end{align}
where $\epsilon = M/N - 1$.  

\section{Terminating the Contention Period in Frameless ALOHA}
\label{sec:threshold}

The central feature of frameless ALOHA is that the contention period is not fixed a priori, but contention is terminated adaptively, aiming to maximize the achieved throughput.
Consequently, the aim is to optimize both the target degree $G$ and the parameters related to the termination of the contention period.
In this section we elaborate the guiding principles for terminating the contention period.
We develop analysis for the case when the noise can be neglected, see \eqref{eq:composite}, and asses the impact of noise in Section~\ref{sec:noise}.

In a typical rateless coding scenario, the criterion for terminating the transmission of encoded symbols is when the complete message has been decoded on the receiving end.
However, due to the constraints of the proposed framework, an identical criterion, where the contention period terminates upon resolving all users, would lead to inefficient use of system resources (i.e., slots) and low throughput.
Particularly, from \eqref{eq:puser} stems that the probability of a user being of degree 0 (i.e., not transmitting at all) is:
\begin{align}
\label{eq:zero}
P \left[ |u| = 0 \right] = \Lambda_0 = e^{-( 1 + \epsilon ) G} = e^{-\frac{M}{N} G}
\end{align}
and it exponentially decays with $M$, with the decay constant $\frac{G}{N}$.
As shown in \cite{SPV2012}, and also due to the limitations of SIC \cite{FJPSL2011,ZZ2012}, the interesting values of $G$ are rather low, implying that the probability of user not transmitting and, therefore, not even having a chance to be resolved, decreases rather slowly with $M$.
In other words, waiting for all users to become resolved would lead to prohibitively long contention periods. 

On the other hand, our aim is to maximize the slot utilization and the expected throughput $T$.
The key observation is that in coded random access it is not vital to resolve all the users in a single contention period, as the unresolved users can always be directed towards a newly initiated contention.
Therefore, the contention period should ideally be terminated when the instantaneous throughput $T_I$ reaches its maximal value, postponing the unresolved users to a future contention period.
If the contention is terminated at the $M$-th slot and the number of resolved users in the same time is $N_R$, then $T_I$ can be computed as:
\begin{align}
\label{eq:T-I}
T_I = \frac{N_R}{M}.
\end{align}

\subsection{Asymptotic Analysis}

\begin{figure}[tbp]
	\begin{center}
\includegraphics[width=0.85\columnwidth]{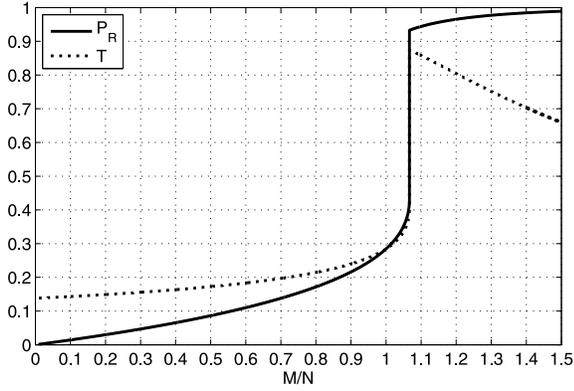}
	\end{center}
\caption{Asymptotic performance of the proposed scheme for $G = 3.12$, obtained using and-or tree evaluation based on \eqref{eq:slot} and \eqref{eq:user}.}
	\label{fig:and-or}
\end{figure}

The asymptotic behavior, when $N \rightarrow \infty$, of the probability of user resolution $P_R$ and the expected throughput $T$:
\begin{align}
T = \frac {P_R}{M/N}= \frac{P_R}{1 + \epsilon},
\end{align}
as functions of the ratio of the number of elapsed slots and number of contending users $M/N$, is shown in Fig.~\ref{fig:and-or}. 
The plotted results are obtained by and-or tree evaluation\cite{LMS1998} using edge-oriented degree distributions derived from \eqref{eq:slot} and \eqref{eq:user}; for details, we refer the interested reader to our previous work \cite{SPV2012}, also noting that the asymptotically optimal value of the target degree $G$ is taken from it.
Obviously, due to the well known avalanche effect, characteristic for the iterative BP erasure-decoding \cite{M2005} (i.e., SIC in our setting), there is a vertical increase in $P_R$ and $T$ at $M/N \approx 1.07$.
At this point, $T$ reaches its maximum, $T \approx 0.874$, and a corresponding termination criterion would be to detect the throughput maximum and end the contention period.

On the other hand, at the point where the avalanche occurs, SIC drives the fraction of resolved users to surge from $P_R = 0.43$ to $P_R = 0.93$ in a single execution cycle.
If one opts to terminate the contention if $P_R \geq V$ and checks this condition after every SIC execution cycle, then setting $V = 0.43 + \delta$, where $\delta \in (0, 0.5)$, should result in throughput that is asymptotically optimal.
The reason is that any value of $\delta$ within the specified range enables the capture of the immediate rise both in $P_R$ and $T$, when $T$ hits its maximal value, as depicted in Fig.~\ref{fig:and-or}.

However, when evaluating the performance of iterative BP decoding, the asymptotic results are not readily transferable to non-asymptotic scenarios, and heuristic/simulation based approaches are typically used.
In the next subsection, we further elaborate the behavior of the scheme and investigate optimized stopping criteria for non-asymptotic number of users using simulation tools.

\subsection{Non-Asymptotic Analysis}
\label{sec:NAA}

\begin{figure}[t]
	\begin{center}
\includegraphics[width=0.85\columnwidth]{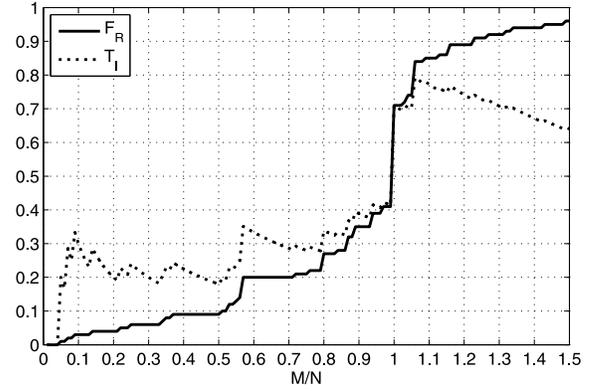}
	\end{center}
\caption{Example of a typical, non-asymptotic performance of the proposed scheme, $N=100$, $G=2.68$.}
	\label{fig:typical}
\end{figure}

Fig.~\ref{fig:typical} depicts a sample evolution of the fraction of resolved users $F_R = \frac{N_R}{N}$ and instantaneous throughput $T_I$ as the number of slots increase, when $N=100$ and $G=2.68$.
As it can be observed, in contrast to the asymptotic case, $T_I$ exhibits several local maxima apart from the global maximum, while the avalanche effect in increase of $T_I$ is not as distinct.
Also, the value of the global maximum in general depends on the particular instance of the SIC evolution.
For example, if the contention period starts with a singleton slot, it is optimal to terminate the contention period immediately after that, i.e., only after a single resolved used, as in this trivial case the system attains $T_I = 1$.
Consequently, the termination criterion that is based solely on $T_I$ and that is able to detect the global maximum as soon as it happens, would be substantially more involved to implement in the non-asymptotic case.
On the other hand, $F_R$ shows a stable performance, in the sense that it monotonically increases with the number of elapsed slots, see the example in Fig.~\ref{fig:typical}; this feature could be exploited when devising a practical criterion for contention termination.

In the next subsection we explore the performance of a heuristic, threshold-based criterion, which consists of \emph{two conditions} - the contention period is terminated either when $T_I \geq S$ or $F_R \geq V$, where $S$ and $V$ are the respective thresholds.
In order to define an upper bound on the performance, we also devise a \emph{genie-aided termination criterion}: the receiver knows non-causally the throughput that will be obtained at all future slots and terminates the contention at the moment in which the throughput is maximal.

\subsection{Results}

All presented results are obtained by averaging 10000 simulation runs for each set of parameters $N$, $G$, $S$ and $V$.
Also, only the features of the proposed access scheme were taken into account in the performed simulations, implementing SIC as a standard iterative BP erasure-decoder \cite{L2002} and neglecting the impact of the physical layer.
However, despite the simplicity of the approach, the presented results can serve as reliable guidelines of the performance achievable in practical scenarios, as justified in \cite{L2011}.

Table~\ref{tab:frameless} compares the genie-aided performance and the performance of the proposed scheme,\footnote{The presented values are rounded to the first two decimals.} for varying number of users $N$.
The throughput values for genie-aided method $\bar{T}_{GA}$ are obtained in the following way.
For each value of $N$, the target degree $G$ is varied with a step $\delta=0.01$ through the candidate range of values.
For each value of $G$, in each simulation run the maximal throughput over $M=10N$ is recorded, and then averaged over runs.
$\bar{T}_{GA}$ is the maximum of the averages and serves as a benchmark.

The maximum average throughput of the scheme with the two condition-based termination criterion (as introduced at the end of the previous subsection) is denoted by $\bar{T}^*$ and obtained in the similar way.
For each $N$ a maximization is performed over parameters $G$, $S$ and $V$, which are varied with a step $\delta=0.01$ through the range of interest.
For every combination of $G$, $S$ and $V$, each simulation run is executed until $F_R \geq V$ or $T_I \geq S$, when $F_R$, $T_I$ and the number of slots $M$ are recorded and then averaged over runs.
$\bar{T}^*$ is the maximum of the throughput averages; the optimal values of the parameters that yield $\bar{T}^*$ are listed as $G^*$, $S^*$ and $V^*$ in Table~\ref{tab:frameless}.
Finally, Table~\ref{tab:frameless} also presents $\bar{F}_{R}$, $\bar{M}/N$, and $\bar{R}$, which are the average fraction of resolved users, the normalized average number of slots in the contention period, and the average number of transmitted replicas per user\footnote{It is straightforward to show that $\bar{R} = \frac{\bar{M}}{N} G$.}, respectively, obtained for $G^*$, $S^*$ and $V^*$.

Comparing $\bar{T}^*$ with $\bar{T}_{GA}$, it could be noted that the proposed termination criterion leads to throughputs that approach the benchmark values.
More importantly, despite the simplicity of the approach, the obtained throughputs for number of users in the presented range of $50-1000$ are exceptionally high, and, to the best of our knowledge, unmatched in the state-of-the-art for the assumed communication model \cite{L2011,SPV2012}.

The optimal target slot degree $G^*$ modestly grows with $N$, starting from $G^* = 2.68$ for $N=50$ and increasing to $G^* = 3.03$ for $N=1000$.
The same behavior was observed in \cite{SPV2012}, where it was shown that $G=3.12$ is asymptotically optimal.
The optimal threshold on the fraction of resolved users $V^*$ also modestly grows with $N$, while the optimal throughput threshold $S^*$ is constant and equal to its maximal value, i.e., $S^* = 1$.
Specifically, the simulation study showed that performance is more sensitive to the choice of $V$ than $S$; once $V$ is properly selected, a simple choice of $S^*=1$ maximizes the expected throughput.

\renewcommand{\arraystretch}{1.3}
\begin{table}[t]
	\centering
		\begin{tabular}{|c||c|c|c|c|}
			\hline
			$N$ & 50 & 100 & 500 & 1000 \\
			\hline
			\hline
 			$\bar{T}_{GA}$ & 0.83 & 0.84 & 0.88 & 0.88 \\			
			\hline
			\hline
      $\bar{T}^*$ & 0.82 & 0.84 & 0.87 & 0.88 \\
			\hline
			$\bar{F}_{R}$ & 0.75 & 0.76 & 0.76 & 0.76 \\
			\hline
			$\bar{M} / N$ & 0.97 & 0.95 & 0.9 & 0.9 \\
			\hline
			$\bar{R}$ & 2.6 & 2.69 & 2.69 & 2.73 \\
			\hline
		 	$G^*$ & 2.68 & 2.83 & 2.99 & 3.03  \\
			\hline
			$S^*$ & 1 & 1 & 1 & 1 \\
			\hline
			$V^*$ & 0.83 & 0.87 & 0.88 & 0.89 \\
			\hline
			\hline
			$\bar{T}$ for $G=2.9$, $V^*$ and $S^*$ & 0.81 & 0.84 & 0.87 & 0.87 \\
			\hline
			$\bar{T}$ for $G=2.9$ $V=0.8$ and $S^*$& 0.81 & 0.83 & 0.86 & 0.87 \\
			\hline
		\end{tabular}
	\caption{Performance results of the proposed scheme.}
	\label{tab:frameless}
\end{table}

The results corresponding to the expected fraction of resolved users $\bar{F}_R$ show that the proposed scheme on average resolves roughly 75\% of users; also, it is always $\bar{F}_{R} \leq V^*$, which may seem counterintuitive.
This is due to the fact that for the cases when the contention is terminated when the instantaneous throughput $T_I$ hits $S^*=1$, the termination typically happens after a single slot (which happened to be a singleton slot), when $F_R = \frac{1}{N}$.\footnote{We emphasize that the proposed scheme aims to maximize the use of the system resources, i.e., to maximize the throughput, and not to the resolve all users in a single contention period.
The unresolved users simply continue to contend in the subsequent contention periods, which is a feature typical for any ALOHA-based scheme with backlogged users.
The implementation of the ``complete'' contention procedure in the frameless framework, executed until all users from a batch arrival have been resolved and in which the BS uses the information from all the slots from all contention periods for SIC is investigated in \cite{STPP2013}.}
The same fact can be observed by inspecting the values of $\bar{M}$ - the average duration of the contention period in slots is always lower than $N$, due to the instances when the contention is terminated because $T_I$ reached 1.  
Finally, the average number of transmissions per user $\bar{R}$ slightly increases with $N$, but always remains lower than the average number of transmissions per slot $G^*$, due to the fact that $\frac{\bar{M}}{N} < 1$.

At this point we turn back to the fact that the obtained range of the optimal values of $G$ hint that a choice of constant value that does not depend on $N$, could yield a satisfactory performance.
This is verified in the penultimate column of Table~\ref{tab:frameless}, which indicates that $G$ could be set to an optimized constant $G=2.9$ with no substantial performance loss.
As a further insight into the scheme's robustness to the choice of parameters, the last column of Table~\ref{tab:frameless} presents the throughput results when both $G$ and $V$ are constant and set to 2.9 and 0.8, respectively.
Obviously, the price to pay for the selection of suboptimal $G$ and $V$ is modest, but now neither of the parameters depends on $N$.
Note that further results on the sensitivity of the scheme's performance to selection of $G$ and $V$, when the knowledge of $N$ is imperfect, are presented in Section~\ref{sec:sensitivity}.

\subsection{Impact of Noise-induced Packet Erasures}
\label{sec:noise}

The operational analogy with rateless coding grants another advantage over related framed ALOHA-based schemes \cite{L2011,NP2012,PLC2011}.
Namely, the slots that become garbled due to noise can be simply neglected, in the same way as erroneously received symbols are discarded in the rateless coding scenario.
In order to see this, assume that the probability that a packet is erased whenever a user is transmitting in a singleton slot is $P_e$; since the baseband precoding is used, each user has the same received SNR and therefore $P_e$ is equal for all users.
We state the following:

\begin{proposition}
\label{proposition}
If the probability of noise-induced erasure in a singleton slot is $P_e$, the expected throughput reduces to:
\begin{align}
\label{eq:noisy}
T' = T (1 - P_e),
\end{align}
where $T$ is the noiseless expected throughput. 
\end{proposition}
\begin{proof}
Denote by $M'$ the number of observed slots and by $M_u$ the expected number of slots that are useful in the SIC process; note that $M_u$ determines the expected number of resolved users $\bar{N}_R$.
In order to characterize the expected number of useful slots $M_u$ from $M'$, we first observe that zero-degree slots are by default useless and do not contribute to $M_u$.
The probability that a singleton slot will be useful in decoding (i.e., packet not erased) is obviously $1-P_e$.
Let's assume that $s_j$ is a collision slot and the corresponding received signal is $Y_j = \sum_{k \in A_j}^{}X_k + Z_j$, where $A_j$ is the set of users $u_k$ for which $a(k,j)=1$, see \eqref{eq:composite}.
Due to the assumed perfect SIC, the cancellation of already resolved transmissions from $Y_j$ does not depend on the slot degree $s_j$.
Thus, the degree of $s_j$ potentially can be reduced to one, when its usability again becomes $1 - P_e$.
This simple analysis reveals that $M_u = ( 1 - \Psi_0)( 1 - P_e) M'$, where $\Psi_0$ is the probability that slot degree is zero, see \eqref{eq:binom}. 
On the other hand, in the noiseless case, the same expected number of useful slots will be obtained for the expected number of observed slots equal to
\begin{align}
M = \frac{M_u}{1 - \Psi_0} = M' ( 1 - P_e )
\end{align}
Substituting $T = \frac{\bar{N}_R}{M} $ in \eqref{eq:noisy} proves the proposition.
\end{proof}

We note that the performed simulations indeed verified the correctness of the proposition, as we the average throughput in the noisy case scaled down as given in \eqref{eq:noisy}.
We also note that, when there is no baseband precoding \eqref{eq:precoding}, the previous proposition cannot be used.
In this case, the contribution of the user $u_i$ to the received composite signal is $g_i X_i$, where $g_i$ is a user-dependent coefficient.
The probability of erasure for a slot depends on which user is the last one left unresolved, as each user has, in general, a different packet erasure probability that is determined by the actual value of $g_i$.

\section{Imperfect Knowledge of $N$}
\label{sec:sensitivity}

In order to achieve the optimal performance in any ALOHA-based random access scheme, the scheme's parameters should be selected according to the number of contending users $N$.
For instance, in framed ALOHA and related schemes, the optimal number of slots $M$ that maximizes $T$ depends on $N$ \cite{R1975,CGH2007,L2011}.

In the proposed approach, both the slot-access probability $p_a$ and the evaluated fraction of the resolved users $F_R$ depend on $N$.
The value of $p_a$ is given by \eqref{eq:access} such that the desired slot degree distribution is obtained.
On the other hand, the contention period is terminated when $F_R = \frac{N_R} {N} \geq V$.
In other words, $N$ influences both the evolution of SIC through $p_a$ and its proper termination through $F_R$; therefore, the BS has to have its estimate before the contention period starts.\footnote{For the sake of completeness, we note that the instantaneous throughput $T_I$ does not depend on $N$, see \eqref{eq:T-I}. Also, we note that $S=1$ for the rest of the section.}
The problem of estimating $N$ within the frameless framework was addressed in \cite{STPP2013}.
We continue the paper by presenting the results on how accurate the estimate should be in order to achieve satisfactory performance.
Let $N_{\text{est}}$ denote the estimate of the actual number $N$ of contending users.
We assume that $N_{\text{est}} = ( 1 + \alpha) N$, where $\alpha$ is the relative estimation error.
From \eqref{eq:access} it follows:
\begin{align}
p_a = \frac{G}{N_{\text{est}}} = \frac{G_{\text{act}}}{N},
\end{align}
where $G_{\text{act}}$ is the actual target degree:
\begin{align}
G_{\text{act}} = \frac{G}{1 + \alpha}.
\end{align}
Following the same reasoning, it could be shown that:
\begin{align}
\label{eq:V}
V_{\text{act}} = \frac{V}{1 + \alpha},
\end{align}
where $V$ and $V_{\text{act}}$ are the target and the actual threshold on the fraction of resolved users.\footnote{
A general condition that has always to be satisfied is $V_{\text{act}} \leq 1$, including also the negative values of $\alpha$.}

\begin{figure}[tbp]
	\begin{center}
\includegraphics[width=0.93\columnwidth]{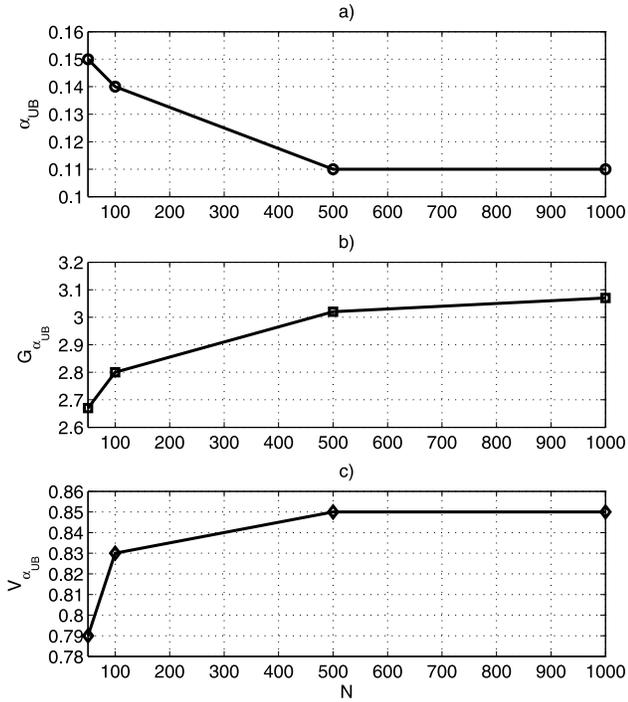}
	\end{center}
\caption{Allowing throughput loss of 5 \%, the figure depicts the corresponding: (a) $\alpha_{UB}$ - the upper bound on $\alpha_{\max}$, which defines the range of the acceptable relative estimation error, (b) target degree $G_{\alpha_{UB}}$ and (c) threshold on the fraction of resolved users $V_{\alpha_{UB}}$.}
	\label{fig:sensitivity}
\end{figure}

The optimization of the scheme's parameters $G$ and $V$ should be such that $E[T_{\max}(N) - T(N,G,V,\alpha)]$ is minimized; this requires both the analytical model of the estimator and $T=T(N,G,V,\alpha)$, which are out of the scope of the paper.
Henceforth, our only assumption is that the estimate is unbiased and that $\alpha \in [-\alpha_{\max}, \alpha_{\max}]$, and we investigate how large $\alpha_{\max}$ can be, given the allowed performance loss.

Fig.~\ref{fig:sensitivity} shows the upper bound $\alpha_{UB}$ on $\alpha_{\max}$ as function of the number of contending users $N$, when the allowed throughput loss is at most 5\% of its overall maximum value $\bar{T}_{GA}$, as listed in Table~\ref{tab:frameless}.
The figure also shows the corresponding target degree $G_{\alpha_{UB}}$ and threshold on fraction of resolved users $V_{\alpha_{UB}}$, for which this performance is achieved.
It can be observed that $\alpha_{UB} \geq 0.11$, i.e., the range of the acceptable estimation error is at least $[-0.11 N, 0.11 N]$.
When $G_{\alpha_{UB}}$ and $V_{\alpha_{UB}}$ are compared to optimal $G^*$ and $V^*$ (given in Table~\ref{tab:frameless}), it could be inferred that their behavior follows the same trends; also, $V_{\alpha_{UB}}$ is noticeably lower than $V^*$, which is due to \eqref{eq:V}. 
We also observed that the expected throughput obtained for $G_{\alpha_{UB}}$ and $V_{\alpha_{UB}}$ when there is no estimation error are negligibly lower (less than $10^{-2}$) than the ones obtained for the optimal $G^*$ and $V^*$, as demonstrated in Table~\ref{tab:subopt}.
In other words, choosing $G_{\alpha_{UB}}$ and $V_{\alpha_{UB}}$ instead $G^*$ and $V^*$ does not adversely affect performance.

\section{Discussion}
\label{sec:discussion}

\subsection{Practical Considerations}
\label{sec:practical}

The scheme discussed in this paper uses several idealized assumptions.
One of those, typical for coded random access schemes, is that each replica contains pointers to the locations of the other replicas transmitted by the same user. However, as the frame length is not fixed, it is neither trivial to make the pointers nor the cost of sending that many pointers is negligible.
A more elegant approach to solve the pointer problem could be the following.
Let us assume that the contention period starts by having the BS broadcasting a beacon that contains a random bit string, denoted by $B$.
The user with address $u_i$ uses a pseudorandom generator in order to determine in which slots $s_j$ to transmit.
Specifically, there is a function $f(u_i,B,s_j)$ that determines the values of $a(i,j)$, see \eqref{eq:composite}, after the beacon value $B$ has been sent.
The function is defined in a way that in a fraction of $p_a$ of the slots, $f$ gets a value of $1$.
Assuming that each replica of the user contains its address $u_i$, then after a replica is resolved, the BS can use the function $f$ and the knowledge of $B$ to find the slots in which the user has transmitted.

\renewcommand{\arraystretch}{1.3}
\begin{table}[t]
	\centering
		\begin{tabular}{|c||c|c|c|c|}
			\hline
			$N$ & 50 & 100 & 500 & 1000 \\
			\hline
			\hline
 			$G_{\alpha_{UB}}$ & 2.67 & 2.8 & 3.02 & 3.07 \\			
			\hline
			\hline
      $V_{\alpha_{UB}}$ & 0.79 & 0.83 & 0.85 & 0.85 \\
			\hline
			$\bar{T}_{\alpha_{UB}}$ & 0.82 & 0.84 & 0.87 & 0.88 \\
			\hline
		\end{tabular}
	\caption{Throughput performance for $G_{\alpha_{UB}}$ and $V_{\alpha_{UB}}$, when there is no estimation error.}
	\label{tab:subopt}
\end{table}

A practical problem that is specific to the frameless operation is the one of terminating the contention process.
In case of a FDD system or a system in which uplink and downlink transmissions are separated in time, like in cellular systems, beacons arrive in a separate channel to the users.  
On the other hand, if uplink and downlink share the same time and frequency, then the BS contends with the users for the link time when transmitting; we investigate the impact of such operation on the performance of the scheme in the further text.

Assume that after the $i$-th slot the BS decides to terminate the contention process by broadcasting a new beacon, piggybacking on it the acknowledgments to all the users that have been resolved since the previous beacon.
Further, assume that the beacon is transmitted with more power than user packets, effectively capturing the channel.
This way, only the users that decided to transmit in slot $i+1$ would collide with the beacon, miss the start of the new contention period and continue according to the old contention pattern, irrespective of whether or not they have been resolved.
In order to mitigate this problem, the beacon sent by the BS should occupy $L$ slots, such that with probability:
\begin{align}
\label{eq:}
P_{\text{miss}} = p_a^L = \frac{G^L}{N^L},
\end{align}
a user misses that a new beacon has been sent.
Taking into account that $G \approx 3$, see Tables~\ref{tab:frameless} and \ref{tab:subopt}, $P_{\text{miss}}$ can be made close to 0 for low values of $L$, even for a low number of contending users $N$.

On the other hand, the prolonged beacon could adversely impact the throughput.
Particularly, the instantaneous throughput could be now computed as:
\begin{align}
\label{eq:longer_beacon}
T_I = \frac{N_R}{M + (L - 1) },
\end{align}
where the term $L - 1$ accounts for the slots that could have been used for the contention instead.
If the contention ends while the current number of elapsed slots $M$ is rather low, which is typically the case when $T_I$ hits its threshold $S^*=1$,\footnote{For instance, if $G=3$ is assumed, then $M=1$ when contention ends in approximately 15\% of the cases, see \eqref{eq:binom}.} $L-1$ is comparable with $M$ and lowers the expected throughput substantially.
In this case, a better strategy is to reduce the termination criterion only to the condition $F_R \geq V$, abandoning the condition $T_I = 1$.
The reason for this is that the fraction of the resolved users $F_R$ reaches or exceeds its threshold $V$ when $M$ is comparable with $N$, decreasing the impact of $L-1$ in \eqref{eq:longer_beacon}.

For example, assume that beacon length is $L=3$ slots.
When $L=3$ and $N=50$, it could be shown that $P_{\text{miss}} < 2 \cdot 10^{-4}$; for larger $N$, $P_{\text{miss}}$ decreases exponentially.
Table~\ref{tab:beacon} shows the maximized expected throughput $\bar{T}$ for $L=3$ when only the criterion $F_R \geq V$ is used to stop the contention, as well as the corresponding values for $G$ and $V$ that maximize it.
Clearly, there is a throughput loss due to the prolonged beacon and usage of the reduced termination criterion.
In the worst case, when $N=50$, the throughput loss is about 7\% when compared to the best possible performance (i.e., $\bar{T}_{GA}$) presented in Table~\ref{tab:frameless}; as $N$ increases the throughput loss decreases.
Nonetheless, the throughput results presented in Table~\ref{tab:beacon} are still better than the ones provided in the previous literature for the given range of $N$.

\renewcommand{\arraystretch}{1.3}
\begin{table}[t]
	\centering
		\begin{tabular}{|c||c|c|c|c|}
			\hline
			$N$ & 50 & 100 & 500 & 1000 \\
			\hline
			\hline
			$\bar{T}$ & 0.76 & 0.8 & 0.85 & 0.86 \\
			\hline			
 			$G$ & 2.85 & 2.89 & 3.02 &  3.08 \\			
			\hline
			\hline
      $V$ & 0.87 & 0.85 & 0.89 & 0.9 \\
			\hline
		\end{tabular}
	\caption{Maximized throughput performance $\bar{T}$ when beacon length is $L=3$ and only the condition $F_R \geq V$ is used to terminate the contention, and corresponding $G$ and threshold $V$.}
	\label{tab:beacon}
\end{table}

Finally, we note that if the beacon is not decoded correctly, the user can stay quiet until the next beacon; other concerns related the design of the beacon require a detailed analysis that is out of the scope of the paper.

\subsection{Further Insights in the Relations to Rateless Codes}

As outlined in the text, the inspiration for the proposed scheme was drawn from rateless coding framework.
However, in contrast to rateless coding scenario, the output (slot) degree distribution cannot be controlled directly, and the slot degrees are Poisson-distributed with mean value equal to the target degree $G$, see \eqref{eq:binom}.
In other words, slot degree distributions that look like robust soliton distribution cannot be achieved.
LT codes avoid the coupon collector problem, expressed in \eqref{eq:zero}, by scaling the mean $G$ of the output distribution as $O(\log N)$, where $N$ is the number of input symbols i.e. users.
The results presented in Table~\ref{tab:frameless} show that the optimal mean $G^*$ of the output degree distribution modestly increases with $N$ and that $G$ could be approximated by a constant with no substantial performance loss.
This is similar to the raptor coding framework \cite{S2006}, where the mean degree of the output distribution is constant.
Raptor codes fight the coupon collector problem by exploiting a high-degree pre-code, which is not applicable in our case due to: (1) a high targeted value of $G$ would
mean that the variance of the actual slot degree would also increase proportionally, as the mean and variance for the Poisson distribution are
equal, and the precoding would thus not be effective, and (2) the slots with high degrees would adversely affect the performance of the SIC algorithm in practice \cite{FJPSL2011,ZZ2012}.
In the proposed scheme, we avoid the coupon collector problem by terminating the contention period such that the throughput is maximized, disregarding the fact that there are users that remain unresolved.
These users continue to contend in the subsequent rounds, in the same way as in any ALOHA-based scheme.

\section{Conclusions}
\label{sec:conclusion}

This paper has introduced the random access scheme termed frameless ALOHA, in which the contention period is adaptively terminated in order to maximize the throughput. Although based on rather simple principles, the proposed scheme provides considerably high throughputs, in fact the highest reported so far for low to moderate number of users.
In other words, a strategy in which the length of the contention period is adaptively tuned to the evolution of the SIC provides a superior throughput performance in comparison to the strategies in which the frame length is fixed.
Rateless-like behavior provides for yet another advantage, as the slots garbled by noise can be simply discarded; i.e., the proposed approach has an inherent potential to adapt not just to the evolution of SIC, but also to the wireless link conditions.

Finally, as the further work we are investigating the ways to boost the intermediate performance of the scheme; an approach where the users are divided into classes with different slot access probabilities seems as a promising direction in this respect.
This approach also achieves different probability of user resolution, which can be of interest in practical applications where certain classes of users report more important data.
Another extension is related to inclusion of the capture effect into analysis; preliminary research shows that in this case higher target degrees are favored and higher throughputs can be achieved. 



\begin{thebibliography}{10}
\providecommand{\url}[1]{#1}
\csname url@samestyle\endcsname
\providecommand{\newblock}{\relax}
\providecommand{\bibinfo}[2]{#2}
\providecommand{\BIBentrySTDinterwordspacing}{\spaceskip=0pt\relax}
\providecommand{\BIBentryALTinterwordstretchfactor}{4}
\providecommand{\BIBentryALTinterwordspacing}{\spaceskip=\fontdimen2\font plus
\BIBentryALTinterwordstretchfactor\fontdimen3\font minus
  \fontdimen4\font\relax}
\providecommand{\BIBforeignlanguage}[2]{{%
\expandafter\ifx\csname l@#1\endcsname\relax
\typeout{** WARNING: IEEEtran.bst: No hyphenation pattern has been}%
\typeout{** loaded for the language `#1'. Using the pattern for}%
\typeout{** the default language instead.}%
\else
\language=\csname l@#1\endcsname
\fi
#2}}
\providecommand{\BIBdecl}{\relax}
\BIBdecl

\bibitem{R1975}
L.~G. Roberts, ``Aloha packet system with and without slots and capture,''
  \emph{SIGCOMM Comput. Commun. Rev.}, vol.~5, no.~2, pp. 28--42, Apr. 1975.

\bibitem{OIN1977}
H.~Okada, Y.~Igarashi, and Y.~Nakanishi, ``Analysis and application of framed
  {ALOHA} channel in satellite packet switching networks - {FADRA} method,''
  \emph{Electronics and Communications in Japan}, vol.~60, pp. 60--72, Aug.
  1977.

\bibitem{CGH2007}
E.~Cassini, R.~D. Gaudenzi, and O.~del Rio~Herrero, ``{C}ontention {R}esolution
  {D}iversity {S}lotted {ALOHA} {(CRDSA)}: {A}n {E}nhanced {R}andom {A}ccess
  {S}cheme for {S}atellite {A}ccess {P}acket {N}etworks,'' \emph{IEEE Trans.
  Wireless Commun.}, vol.~6, no.~4, pp. 1408--1419, Apr. 2007.

\bibitem{CR1983}
G.~L. Choudhury and S.~S. Rappaport, ``{D}iversity {ALOHA} - {A} {R}andom
  {A}ccess {S}cheme for {S}atellite {C}ommunications,'' \emph{IEEE Trans.
  Commun.}, vol.~31, no.~3, pp. 450--457, Mar. 1983.

\bibitem{L2011}
G.~Liva, ``{G}raph-{B}ased {A}nalysis and {O}ptimization of {C}ontention
  {R}esolution {D}iversity {S}lotted {ALOHA},'' \emph{IEEE Trans. Commun.},
  vol.~59, no.~2, pp. 477--487, Feb. 2011.

\bibitem{NP2012}
K.~R. Narayanan and H.~D. Pfister, ``Iterative {Co}llision {R}esolution for
  {S}lotted {ALOHA}: {A}n {O}ptimal {U}ncoordinated {T}ransmission {P}olicy,''
  in \emph{Proc of 7th ISTC}, Gothenburg, Sweden, Aug. 2012.

\bibitem{BLMR1998}
J.~Byers, M.~Luby, M.~Mitzenmacher, and A.~Rege, ``{A} {D}igital {F}ountain
  {A}pproach to {R}eliable {D}istribution of {B}ulk {D}ata,'' in \emph{Proc. of
  ACM SIGCOMM 1998}, Vancouver, BC, Canada, Sep. 1998.

\bibitem{SPV2012}
C.~Stefanovic, P.~Popovski, and D.~Vukobratovic, ``{F}rameless {ALOHA} protocol
  for {W}ireless {Networks},'' \emph{IEEE Comm. Letters}, vol.~16, no.~12, pp.
  2087--2090, Dec. 2012.

\bibitem{STPP2013}
C.~Stefanovic, K.~F. Trilingsgaard, N.~K. Pratas, and P.~Popovski, ``{J}oint
  {E}stimation and {C}ontention-{R}esolution {P}rotocol for {W}ireless {R}andom
  {A}ccess,'' in \emph{Proc. of IEEE ICC 2013}, Budapest, Hungary, Jun. 2013.

\bibitem{PLC2011}
E.~Paolini, G.~Liva, and M.~Chiani, ``{H}igh {T}hroughput {R}andom {A}ccess via
  {C}odes on {G}raphs: {C}oded {S}lotted {ALOHA},'' in \emph{Proc. of IEEE ICC
  2011}, Kyoto, Japan, Jun. 2011.

\bibitem{PFC2010}
E.~Paolini, M.~P.~C. Fossorier, and M.~Chiani, ``{G}eneralized and {D}oubly
  generalized {LDPC} {C}odes {W}ith {R}andom {C}omponent {C}odes for the
  {B}inary {E}rasure {C}hannel,'' \emph{IEEE Trans. Inf. Theory}, vol.~56,
  no.~4, pp. 1651--1672, Apr. 2010.

\bibitem{PLC2011b}
E.~Paolini, G.~Liva, and M.~Chiani, ``{G}raph-{B}ased {R}andom {A}ccess for the
  {C}ollision {C}hannel without {F}eed-{B}ack: {C}apacity {B}ound,'' in
  \emph{Proc. of IEEE Globecom}, Houston, TX, USA, Dec. 2011.

\bibitem{LPLC2012}
G.~Liva, E.~Paolini, M.~Lentmaier, and M.~Chiani, ``Spatially-coupled random
  access on graphs,'' in \emph{Proc. of IEEE ISIT 2012}, Boston, MA, USA, Jul.
  2012.

\bibitem{L2002}
M.~Luby, ``{LT} {C}odes,'' in \emph{Proc. of 43rd IEEE FOCS}, Vancouver, BC,
  Canada, Nov. 2002.

\bibitem{FJPSL2011}
K.~Fyhn, R.~M. Jacobsen, P.~Popovski, A.~Scaglione, and T.~Larsen,
  ``{M}ultipacket {R}eception of {P}assive {UHF} {RFID} {T}ags: {A}
  {C}ommunication {T}heoretic {A}pproach,'' \emph{IEEE Trans. Sig. Proc},
  vol.~59, no.~9, pp. 4225 --4237, Sep. 2011.

\bibitem{ZZ2012}
A.~Zanella and M.~Zorzi, ``{T}heoretical {A}nalysis of the {C}apture
  {P}robability in {W}ireless {S}ystems with {M}ultiple {P}acket {R}eception
  {C}apabilities,'' \emph{IEEE Trans. Commun.}, vol.~60, no.~4, pp. 1058--1071,
  Apr. 2012.

\bibitem{PLC2012}
E.~Paolini, G.~Liva, and M.~Chiani, ``{R}andom {A}ccess on {G}raphs: {A}
  {S}urvey and {N}ew {R}esults,'' in \emph{Proc. of IEEE Globecom}, Houston,
  TX, USA, Dec. 2011.

\bibitem{RU2007}
T.~Richardson and R.~Urbanke, \emph{{M}odern {C}oding {T}heory}.\hskip 1em plus
  0.5em minus 0.4em\relax Cambridge University Press, Cambridge, UK, 2007.

\bibitem{LMS1998}
M.~G. Luby, M.~Mitzenmacher, and A.~Shokrollahi, ``{A}nalysis of {R}andom
  {P}rocesses via {A}nd-{O}r {T}ree {E}valuation,'' in \emph{Proc. of 9th
  ACM-SIAM SODA}, San Francisco, CA, USA, Jan. 1998.

\bibitem{M2005}
D.~J.~C. MacKay, ``{F}ountain {C}odes,'' \emph{IEE Proc. Commun.}, vol. 152,
  no.~6, pp. 1062--1068, 2005.

\bibitem{S2006}
A.~Shokrollahi, ``{R}aptor {C}odes,'' \emph{IEEE Trans. Info. Theory}, vol.~52,
  no.~6, pp. 2551--2567, Jun. 2006.

\end{thebibliography}
\end{document}